\documentclass[runningheads]{llncs}
\usepackage{graphicx}
\usepackage{epsfig}
\usepackage{graphics}
\usepackage{array}
\usepackage{amsmath}
\usepackage{algorithm}
\usepackage{algorithmic}
\usepackage{comment}
\usepackage{fullpage}
\date{}

\title{Simpler Sequential and Parallel Biconnectivity Augmentation}
\author{Surabhi Jain and N.Sadagopan} 
\institute{Department of Computer Science and Engineering,\\ Indian Institute of Information Technology, Design and Manufacturing, Kancheepuram, Chennai, India. \\
\email{\{surabhijain,sadagopan\}@iiitdm.ac.in}}
\begin{document}
\maketitle
\begin{abstract}
For a connected graph, a vertex separator is a set of vertices whose removal creates at least two components and a minimum vertex separator is a vertex separator of least cardinality.   The vertex connectivity refers to the size of a minimum vertex separator.  For a connected graph $G$ with vertex connectivity $k~ (k \geq 1)$, the connectivity augmentation refers to a set $S$ of edges whose augmentation to $G$ increases its vertex connectivity by one.  A minimum connectivity augmentation of $G$ is the one in which $S$ is minimum.  In this paper, we focus our attention on connectivity augmentation of trees. Towards this end, we present a new sequential algorithm for biconnectivity augmentation in trees by simplifying the algorithm reported in \cite{nsn}.  The simplicity is achieved with the help of edge contraction tool.  This tool helps us in getting a recursive subproblem preserving all connectivity information.  Subsequently, we present a parallel algorithm to obtain a minimum connectivity augmentation set in trees.  Our parallel algorithm essentially follows the overall structure of sequential algorithm.  Our implementation is based on CREW PRAM model with $O(\Delta)$ processors, where $\Delta$ refers to the maximum degree of a tree.  We also show that our parallel algorithm is optimal whose processor-time product is $O(n)$ where $n$ is the number of vertices of a tree, which is an improvement over the parallel algorithm reported in \cite{hsu}.  
\end{abstract}
\section{Introduction}
Connectivity augmentation is a classical combinatorial optimization that finds application in the study of resilent computer networks.  The parameter connectivity models the {\em robustness} of a network and increasing the robustness by one by adding a minimum number of links is a fundamental problem in the study of reliable computer networks.  The study of connectivity augmentation was initiated by Tarjan et al. in \cite{tarjan} and subsequently it has attracted many researchers in the past.   In \cite{tarjan}, combinatorial study on biconnectivity augmentation of singly connected graphs (graphs with vertex connectivity one) is presented.  Subsequently, Hsu et al. in \cite{hsu} presented a linear-time sequential algorithm for connectivity augmentation of singly connected graphs.  Connectivity augmentation of 2-connected graphs and 3-connected graphs were reported in \cite{a,b}.  Vegh in \cite{vegh} settled the computational complexity of connectivity augmentation of $k$-connected graphs for any $k \geq 1$ which was open for almost 35 years.   The algorithm reported in \cite{vegh} runs in $O(n^7)$.  The complexity of general connectivity augmentation which asks for increasing the vertex connectivity to $l$ given a graph with vertex connectivity $k$ is still open.  From the lower-bound theory perspective, any biconnectivity augmentation needs at least $O(n)$ computations to output a minimum biconnectivity augmentation set.  This fact also implies that the algorithm reported in \cite{tarjan,hsu} are optimal.  One way to speed up sequential computation is to design parallel algorithms.   Since almost all modern day computers are at least Quad core with at least four active threads, it is natural to think of parallel algorithms for combinatorial problems.  The only available result in the literature is for parallel biconnectivity augmentation by Hsu et al. \cite{hsu}.  The algorithm reported in \cite{hsu} uses $O(n)$ processors with each processor takes $O(log^2n)$ time.  Hence, the overall run-time of parallel biconnectivity augmentation reported in \cite{hsu} incurs $O(n.log^2n)$ with CREW strategy.   In parallel algorithmics setting, an algorithm is optimal if its process-time product equals the sequential lower bound.   In this paper, we propose an optimum parallel biconnectivity augmentation algorithm with $O(\Delta)$ processors which is an improvement over \cite{hsu}, where $\Delta$ refers to the maximum degree of a graph.  Our algorithm uses $O(\Delta)$ processors with each processor takes $O(\frac{n}{\Delta})$ time.  To achieve this, we propose a reduction technique for reducing the given graph using the operation {\em edge contraction}.  This reduction tool yields a simple and an elegant sequential and parallel algorithm for biconnectivity augmentation.  We believe that this reduction tool can be used in other combinatorial problems for preprocessing the input graph preserving all its connectivity information.       \\ \\
{\bf Past Results on Edge Contraction:} Edge contraction is a very popular operation in Graph Theory and related areas.  It is one of the two operations used to define the minor of a graph: $H$ is a {\em minor} of $G$ if $H$ can be obtained from $G$ by a sequence of edge deletions or edge contractions.    Forbidden minors are used to obtain characterization of many classes of graphs specified by representation or structural properties. A classic example is the forbidden minor characterization of planar graphs.  Edge contraction and in general clique contraction plays a significant role in the proof of Perfect Graph Theorem \cite{golu}, and in other structural results \cite{matthias}.  The basic idea exploited in randomized algorithms for min-cut is that contracting a randomly chosen edge does not increase the size of the min-cut \cite{karger}.  This leads to expected polynomial time algorithms for min-cut, and these algorithms are fundamentally different from the classical max-flow based techniques.  Edge contraction and vertex connectivity are well studied in the literature, for example, in finding a structural characterization of contractible edges.  In this paper, we use edge contraction as an operation to reduce the size of the graph preserving its connectivity information.  We believe that this reduction tool can be used in other combinatorial problems related to vertex connectivity.  In our work, this tool is quite powerful as for certain inputs it reduces the size of the tree from $O(n)$ to $O(\Delta)$.\\ 
{\bf Our Results:} We first present a simple sequential algorithm for biconnectivity augmentation by modifying the algorithm presented in \cite{nsn} using {\em edge contraction} tool.  This new approach yields a simple implementation although the run-time is still linear in the input size.  We next present a parallel biconnectivity augmentation algorithm using $O(\Delta)$ processors.  Our parallel algorithm is optimal, i.e. the processor-time product is $O(n)$ which is an improvement over the parallel algorithm reported in \cite{hsu}.
\section{Preliminaries}
\subsection{Connectivity Augmentation Preliminaries}
Notation and definitions are as per \cite{golu,west}.  Let $G =(V,E)$ be an undirected connected graph where $V(G)$ is the set of vertices and $E(G) \subseteq \{\{u,v\}~|~ u,v \in V(G)$, $u \not= v \}$. For $v \in V(G)$, $d_G(v)$ refers to the degree of $v$ in $G$.  $\delta(G)$ and $\Delta(G)$ refers to the minimum and maximum degree of $G$, respectively.  For simplicity, we use $\delta$ and $\Delta$ when the associated graph is clear from the context. For $S \subset V(G)$, $G[S]$ denotes the graph induced on the set $S$ and $G \setminus S$ is the induced graph on the vertex set $V(G) \setminus S$.  A vertex separator of a graph $G$ is a set $S \subseteq V(G)$ such that $G \setminus S$ has more than one connected component.  A minimum vertex separator $S$ is a vertex separator  of least size and the cardinality of such $S$ is the vertex connectivity of a graph $G$, written  $\kappa(G)$.  A graph is $k$-vertex connected if $\kappa(G)=k$.  If $\kappa(G)=1$ then the graph is 1-connected (also known as singly connected) and in such a graph a minimum vertex separator $S$ is a singleton set and the vertex $v \in S$ is a {\em cut-vertex} of $G$. We let $G \cdot e$ denote the graph obtained by contracting an edge $e=\{u,v\}$ in $G$ such that $V(G \cdot e)=V(G)\setminus\{u,v\}\cup\{z_{uv}\}$ and $E(G \cdot e)=\{ \{z_{uv},x\} ~|~ \{u,x\}$ or $\{v,x\} \in E(G) \cup \{x,y\} ~|~ x \not= u, y \not= v$ in $E(G) \}$.  For a graph $G$ with $\kappa(G)=k$, a minimum connectivity augmentation set $E_{ca} = \{\{u,v\} ~|~ u,v \in V(G) \mbox{ and } \{u,v\} \notin E(G) \}$ is such that $G$ augmented with $E_{ca}$ ($G'=G + E_{ca}$) is of vertex connectivity $k+1$.  The path between $u$ and $v$ is denoted as $P_{uv}$ and is defined on the vertex set $V(P_{uv})= \{u=u_1,\ldots,u_r=v\}$ such that $E(P_{uv})=\{\{u_i,u_{i+1}\} \in E(G), 1 \leq i \leq r-1\}$.   We use two kinds of paths in our work.  $P1=\{P_{uv} ~|~ P_{uv}=\{u=u_1,\ldots,u_r=v\}$ such that $d_G(u) \geq 3$ and $d_G(v) \geq 3\}$. $P2=\{P_{uv} ~|~ P_{uv}=\{u=u_1,\ldots,u_r=v\}$ such that $d_G(u) \geq 3$ and $d_G(v)=1\}$.  For a tree $T$,  $P_{uv} \cdot uv$ denotes the contraction of $P_{uv}$ into an edge $\{u,v\}$, i.e., $V(T)$ becomes $V(T) \setminus \{u_2,\ldots,u_{r-1}\}$ and $E(T)$ becomes $(E(T) \setminus E(P_{uv})) \cup \{u,v\}$.
\subsection{Parallel Computing Preliminaries}
In this paper, we work with Parallel Random Access Machine (PRAM) Model.  It consists of a set of $n$ processors all connected to a shared memory.  The time complexity of a parallel algorithm is measured using $O($number of processors $\times$ time for each processor$)$.  This is also known as processor-time product.  Access policy must be enforced when two processors trying to Read/Write  into the same cell. This can be resolved using one of the following strategies:
\begin{itemize}
\item Exclusive Read and Exclusive Write (EREW): Only one processor is allowed to read/write a cell
\item Concurrent Read and Exclusive Write (CREW): More than one processor can read a cell but only one is allowed to write at a time
\item Concurrent Read and Concurrent Write (CRCW): All processors can read and write a cell at a time.
\end{itemize}
In our work, we restrict our attention to CREW PRAM model.  For a problem with input size $N$ and $P$ processors,  the speedup is defined as $S_p(N)= \frac{T_1(N)}{T_p(N)}$, where $T_i(N)$ is the time taken on a problem size $N$ with $i$ processors.  The efficiency is defined as $E_p(N)=\frac{S_p(N)}{P}$.
\section{Biconnectivity Augmentation in Trees}
\label{treeaugment}
In this section, we first present a new sequential approach for biconnectivity augmentation in trees.  We discuss lower bound analysis, the sequential algorithm followed by a proof of correctness.  In the subsequent section, we present a parallel algorithm for biconnectivity augmentation.  Our parallel algorithm is a simple one which naturally results from the sequential approach.  Our sequential algorithm runs in $O(n)$ time and the parallel algorithm is optimal with processor-time product is $O(\Delta.\frac{n}{\Delta})$.   Parallel strategy uses $O(\Delta)$ processors with CREW strategy.
\subsection{Biconnectivity Augmentation: A New Sequential Approach}
Given a tree, this section presents a new approach to find an optimum augmenting set which makes the tree biconnected.  We first discuss the lower bound on the size of optimum augmenting set, followed by a new proof of tightness.  Our algorithm uses the operation edge contraction to simplify the input graph and this tool yields a linear time algorithm with a simple implementation. 
This new framework also guarantees a tree at each iteration and hence we obtain a recursive sub-problem efficiently.  For a tree $T$, a {\em representative} is a vertex of degree at least 3 such that it has a leaf as its child.  Let $R=\{R_1,\ldots,R_p\}$ denotes the set of representatives and $L=\{L_1,\ldots,L_p\}$ denotes the set of leaves associated with $R$ such that $L_i$ represents the set of leaves associated with $R_i$.  i.e., $L_i$ contains the set of leaves which have a common parent in $T$. Observe that $L$ partitions the set of leaves in $T$.
\subsection{Lower Bound on Biconnectivity Augmentation in trees}
Given a tree $T$ we now present the lower bound on optimum biconnectivity augmenting set.  It is a well-known fact that in any 2-connected graph, for any pair of vertices, there exists two vertex disjoint paths between them.  This fact is useful in determining the lower bound on the optimum biconnectivity augmenting set.  Let $l$ denote the number of leaves in $T$.  Clearly, to biconnect $T$ we must augment at least $\lceil \frac{l}{2} \rceil$ edges.  Another lower bound is due to the number of components created by removing a cut vertex of $T$.  Note that the number of components created by a cut vertex $x$ of $T$ is precisely the degree of $x$ in $T$.   This shows that in any biconnectivity augmentation of $T$, for each cut vertex $x$, one must find at least $deg_T(x)-1$ new edges in the augmenting set.  Therefore, we must augment at least $\Delta(T)-1$ edges to biconnect $T$.  Therefore, by combining the two lower bounds, the number of edges to biconnect $T$ is at least $\max\{\lceil \frac{l}{2} \rceil, \Delta(T)-1\}$.  This lower bound is indeed tight as shown in \cite{tarjan}.  \\
In the next section, we present a new proof of tightness.  In this proof, we identify representatives for the tree at each iteration satisfying degree constraints and show that adding edges among appropriately chosen leaf pairs naturally results in a recursive sub-problem in which the lower bound value is one lesser.   The main contribution here is the use of edge contraction and the identification of two representatives which consequently guarantees easy construction of the recursive sub-problem.  This new approach yields an elegant algorithm  to compute $E_{min}(G)$, an optimal biconnectivity augmenting set.  This approach is fundamentally different from the results presented in \cite{tarjan,hsu} and similar to the approach presented in \cite{nsn}.   We further prove that the number of edges augmented by our algorithm is precisely the lower bound mentioned in this section.   \\ \\ \\
\begin{algorithm}[h]
\caption{Biconnectivity Augmentation in Trees: {\em tree-augment(Tree T)}} \label{tree-alg1}
\begin{algorithmic}
\IF{there are exactly two leaves $x$ and $y$}
\STATE Add the edge $\{x,y\}$ and return the biconnected graph
\ELSE{}
\STATE{$T_{contract}$=Perform-Path-to-Edge-Contraction($T$)}
\STATE{Compute {\em representatives} in $T_{contract}$}
\IF{there exists exactly one representative}
\STATE{\tt /*$T_{contract}$ is a star */}
\STATE{star-augment($T_{contract}$)}
\ELSE{}
\STATE{\tt /* $T_{contract}$ is not a star */}
\STATE{$T'$=non-star-augment($T$)}
\STATE{star-augment($T'$)}
\ENDIF
\ENDIF
\end{algorithmic}
\end{algorithm}
\begin{algorithm}[h]
\caption{Preprocessing using Edge Contraction {\em Perform-Path-to-Edge-Contraction(Tree T)}} \label{tree-alg2}
\begin{algorithmic}
\IF{there exists a vertex of degree exactly two in $T$}
\STATE{for each path $P_{uv}$ of type $P1$ and $P2$ in $T$, perform $P_{uv} \cdot uv$ and let the new tree be $T'$}
\ELSE{}
\STATE{Return $T$ and exit}
\ENDIF
\STATE{Return $T'$}
\end{algorithmic}
\end{algorithm}
\begin{algorithm}[h]
\caption{Biconnectivity Augmentation in non-star Trees {\tt non-star-augment(Tree T)}} \label{tree-alg3}
\begin{algorithmic}[1]
\STATE{Let $R=\{R_1,R_2,...,R_p\}$ be the set of representatives in $T$ and $D=\{d_1,d_2,...,d_p\}$ denote the degree of the representatives, respectively. Let $L_i$ denote the leaf set associated with $R_i$ and $l_i=|L_i|$}
\STATE{Find two representatives $R_i$ and $R_j$ in $R$ such that their degrees $d_i$ and $d_j$ are the maximum and the second maximum}
\IF{(($d_{i}-l_{i}==1)\&\&(d_{j}-l_{j}==$1))}
\STATE{$t=l_{j}-1$}
\ELSIF{(($d_{i}-l_{i}!=1)\&\&(d_{j}-l_{j}$==1))}
\IF{$(min(l_{i},l_{j})==l_{i})$}
\STATE {$t=l_{i}$}
\ELSE
\STATE{$t=l_{j}-1$}
\ENDIF
\ELSIF{ $((d_{i}-l_{i}==1)\&\&(d_{j}-l_{j}$ !=1))}
\STATE{$t=l_{j}$}
\ELSE
\STATE{$t=min(l_{i},l_{j})$}
\ENDIF
\STATE{Pick $X=\{x_1,\ldots,x_t\} \subset R_i$ and $Y=\{y_1,\ldots,y_t\} \subset R_j$.  Add edges $\{\{x_i,y_i\} ~|~ 1 \leq i \leq t\}$ sequentially.  Remove the set $X$ and $Y$ to get a tree for the next iteration.  Also, update the set $R$ and $D$.}
\end{algorithmic}
\end{algorithm}
\begin{theorem}
Let $T$ be a tree with $l \geq 3$ leaves.  Algorithm {\tt tree-augment()} precisely augments $max(\lceil\frac{l}{2}\rceil, \Delta(T)-1)$.
\end{theorem}
\begin{proof}
If $T$ is a star-like tree, then {\em tree-augment} augments $\Delta(T)-1$ edges.  Otherwise there exists two representatives whose degree is at least 3 in $T$. \\
{\bf Case: $\lceil\frac{l}{2}\rceil \geq \Delta(T)-1$}.  We show by Mathematical Induction on $l$, $|E_{ca}|= \lceil\frac{l}{2}\rceil$.  {\em Base case:} $l \leq 3$.  Clearly, $T$ is star-like tree and any $E_{ca}$ requires two edges.  Therefore, the claim $|E_{ca}|=\lceil\frac{l}{2}\rceil$ follows. {\em Hypothesis:} Assume any tree on less than $l \geq 4$ leaves satisfies our claim. {\em Anchor Step:} Consider a tree on $l \geq 4$ leaves.  Since $l \geq 4$, it can not be a star-like tree as $\lceil\frac{l}{2}\rceil \geq \Delta(T)-1$.  This implies that there exists two vertices of degree at least 3 in $T$.  In particular, there exists two representatives $R_i$ and $R_j$ such that both $L_i$ and $L_j$ are non empty.  By our algorithm, we add an edge between $x \in L_i$ and $y \in L_j$ and remove $x$ and $y$ from $T$.  Let $T'$ be the tree obtained from $T$ by removing $x$ and $y$.  Clearly, $T'$ contains less than $l$ leaves and by the induction hypothesis $E_{ca}$ for $T'$ requires $\lceil\frac{l-2}{2}\rceil$ edges.  For $T$, $E_{ca}$ contains $E_{ca}$ of $T'$ plus the edge $\{x,y\}$.  Therefore, $|E_{ca}|= \lceil\frac{l-2}{2}\rceil +1 = \lceil\frac{l}{2}\rceil - 1 + 1 = \lceil\frac{l}{2}\rceil$.  Therefore, the claim follows for all tree with $l \geq 3 $ leaves. \\ \\
{\bf Case: $\lceil\frac{l}{2}\rceil < \Delta(T)-1$}.  Let $v$ be vertex whose degree is $\Delta(T)$.  We first observe that, $v$ is a representative vertex such that the associated leaf set $L_v$ is non empty.  Suppose not, then clearly the number of subtrees rooted at $v$ is strictly more than $\lceil\frac{l}{2}\rceil$.  Moreover, each subtree contains at least two leaves.  Therefore, the total number of leaves in $T$ is  strictly more than $l$, which is a contradiction.  Therefore, $L_v \not= \phi$.  By our algorithm we choose $v$ as one of the candidate representatives.  After an edge addition,  $\lceil\frac{l}{2}\rceil$ decreases by one irrespective of $\Delta(T)$ which may decrease by one.  In any case, the invariant  $\lceil\frac{l}{2}\rceil < \Delta(T)-1$ is preserved after each iteration and therefore the number of edges augmented in this case is  $\Delta(T)-1$.  Hence, the claim follows. \qed
\end{proof}
\begin{theorem}
For a tree $T$, the graph obtained from the algorithm {\tt tree-augment()} is 2-connected.
\end{theorem}
\begin{proof}
If $T$ is star like tree, then clearly the output of {\em star-augment} is 2-connected as it guarantees two vertex disjoint paths between every pair of vertices in $T$.  When $T$ is not a star like tree, we prove our claim using Mathematical Induction on the number of vertices.  The base case of the induction is when $T$ is a star like tree.  For the hypothesis, we assume that {\em star-augment} algorithm outputs a 2-connected graph for all trees of size less than $n$.  Let $T$ be a tree with $n$ vertices.  Since $T$ is not a star, there exists two representative vertices $R_i$ and $R_j$ such that both $R_i$ and $R_j$ are nonempty.  By our approach an edge is added between an element $x$ of $R_i$ to $y$ of $R_j$ and we obtain a recursive sub problem of size $n-2$.  By the induction hypothesis, a tree on $n-2$ vertices can be made biconnected by augmenting a minimum number of edges and let the resulting graph be $G_{n-2}$.  In $G$, clearly between $x$ and $y$ there are two vertex disjoint paths between $x$ and $y$.  i.e., $x$ and $y$ is an edge in $G$ which gives one path and the other path between $x$ and $y$ exists in $T$ itself.  Moreover, we also have two vertex disjoint paths between $x$ and any vertex $z$ in $T$.  One path between $x$ and $z$ is through the neighbour of $x$ and the other path is using the edge $\{x,y\}$.  This completes the induction and hence the graph output by {\em tree-augment} is 2-connected.   \qed 
\end{proof}
\subsection{Implementation of {\tt tree-augment()} and Run-time Analysis}
This section reports a linear-time implementation of the algorithm {\tt tree-augment()}.  The non-trivial subroutine in tree-augment() is {\tt non-star-augment()}.  We now describe the data structures used to implement non-star-augment().  Subsequently, we argue that it yields a linear-time implementation. 
\begin{itemize}
\item We maintain two Hash tables, the Hash table {\em Representative-Hash-table} (H1) contains an entry for each vertex of degree at least 3 in $T$.  Let $V_3$ denote the set of vertices of degree at least 3 in $T$.  Note that not all elements in $V_3$ are representatives at the start of the algorithm.  An element of $V_3$ which is not a representative in the current iteration may become a representative at a later iteration of the algorithm.  In H1, against each entry, we store the set $L_i$ which is a set of leaves associated with the representative $R_i$.   An empty list is attached to those locations which are not representatives in the current iteration.
\item The other Hash table {\em Representative-Degree-Hash-table}(H2) is used to store the degree of the representatives and is useful in retrieving {\em maximum-degree-representative} and {\em second-maximum-degree-representative} efficiently.   The number of entries in H2 is same as the maximum degree of $T$.  In H2, the location $i$ contains the list of representatives of degree $i$ in $T$.  
\item Given a tree $T$, we first call {\em Perform-Path-to-Edge-Contraction()} routine to contract a path of length at least 3 in $T$ into an edge. i.e., let $P_{uv}=\{u=u_1,u_2,\ldots,u_r=v\}$ denote a path in $T$ such that degree of $u$ and $v$ is at least 3 in $T$ and for every other $u_i$ the degree is exactly 2 in $T$.  The routine {\em Perform-Path-to-Edge-Contraction()} replaces $P_{uv}$ with an edge between $u$ and $v$ to get a tree of smaller size.  The standard graph traversal Depth First Search (DFS) can be used to perform the above task in linear time. 
\item Note that both the tables H1 and H2 requires updation during the execution of the algorithm and the updation happens due to the following scenarios.  Let $x$ and $y$ are two representatives with its degrees being maximum and second maximum.  As per our algorithm we add an edge between an element of $L_x$ to an element of the set $L_y$, where $L_x$ and $L_y$ are the leaves associated with the representative $x$ and $y$, respectively.  Moreover, we also remove $x$ and $y$ from $T$ to get a tree for the next iteration.  In fact we add $t$ edges sequentially as mentioned in the algorithm and remove $2t$ leaves from $T$.  Due to the removal of leaves, the degree of the representatives reduces and when it becomes two it is no longer a representative as per our definition.  i.e., it has exactly one leaf associated with it.  Let $R_i$ be the representative whose degree becomes two during the run of the algorithm.  We observe that the parent of $R_i$ ({\em parent($R_i$)}) in $T$ is a vertex of degree at least 3 and $L_i$ is merged with the set associated with {\em parent($R_i$)} in H1.  This update incurs $O(1)$-time and it is done at most once for each $R_i$.  Since there are $O(n)$ representatives, the above step can be done in $O(n)$ time.
\item When the {\em parent($R_i$)} becomes a representative, we need to make an entry in H1 if it does not exist.  For a representative $R_i$, during the run of the algorithm, it may be the case that the set $L_i$ becomes empty due to edge additions and the degree of $R_i$ is at least 3.  In such a case $R_i$ is no longer a representative and is removed from H1.  
\item Using H2 we update the maximum degree and second maximum degree for subsequent iterations.  For example, for the first iteration of the algorithm $x$ and $y$ are the vertices with its degree maximum and second maximum.  For the second iteration, it may be the case that there are vertices $w$ and $z$ in the location $\Delta(T)$ (location of $x$) of H2.  If not, we search for $w$ in $\Delta(T)-1$, $\Delta(T)-2$, and so on till we find a non-empty location.  Similar strategy is adopted to find the representative with second maximum degree for the next iteration.  Clearly, the time spent for the above operation over all iterations of the algorithm  is $O(n)$.  Therefore, the {\tt tree-augment()} runs in time linear in the input size.
\end{itemize}
\section{Parallel Biconnectivity Augmentation: A Novel Approach}
We present a parallel biconnectivity augmentation algorithm using $O(\Delta)$ processors with CREW strategy.  Our parallel algorithm essentially follows the overall structure of the sequential algorithm.  Moreover, it is clear that the parallel algorithm augments $max(\lceil\frac{l}{2}\rceil, \Delta(T)-1)$ and the resulting graph is 2-connected.
\begin{algorithm}[h]
\caption{Parallel Biconnectivity Augmentation in Trees {\em Parallel-tree-augment(Tree T)}} \label{tree-alg5}
\begin{algorithmic}[1]
\STATE{Perform Parallel Edge Contraction using $\Delta$ processors.  Root the tree at a maximum degree node, say $v$.  Let $\Delta$ processors in parallel explore paths of type $P1$ or $P2$ to be contracted into edges}
\STATE{If the resulting tree after edge contraction routine is a star, then augment $\Delta(T)-1$ edges in parallel using $\Delta$ processors}
\STATE{Otherwise, $T$ is a non-star.  Follow the steps of Algorithm \ref {tree-alg3} to find $t$ pairs of vertices to be augmented.  Since $t \leq \Delta(T)$, use $\Delta$ processors to augment $t$ edges in parallel}
\STATE{Also, remove the pairs for which augmentation is done to get a tree for the next iteration. Update the sets $R$ and $D$ for the next iteration.}
\end{algorithmic}
\end{algorithm}
\subsection{Trace of the Parallel Algorithm (Algorithm \ref{tree-alg5})}
We trace the steps of  Algorithm \ref{tree-alg5} in the Figure \ref{fig:trace}. 
\begin{figure}[htpb]
\begin{center}
\includegraphics[angle=0, scale=0.6]{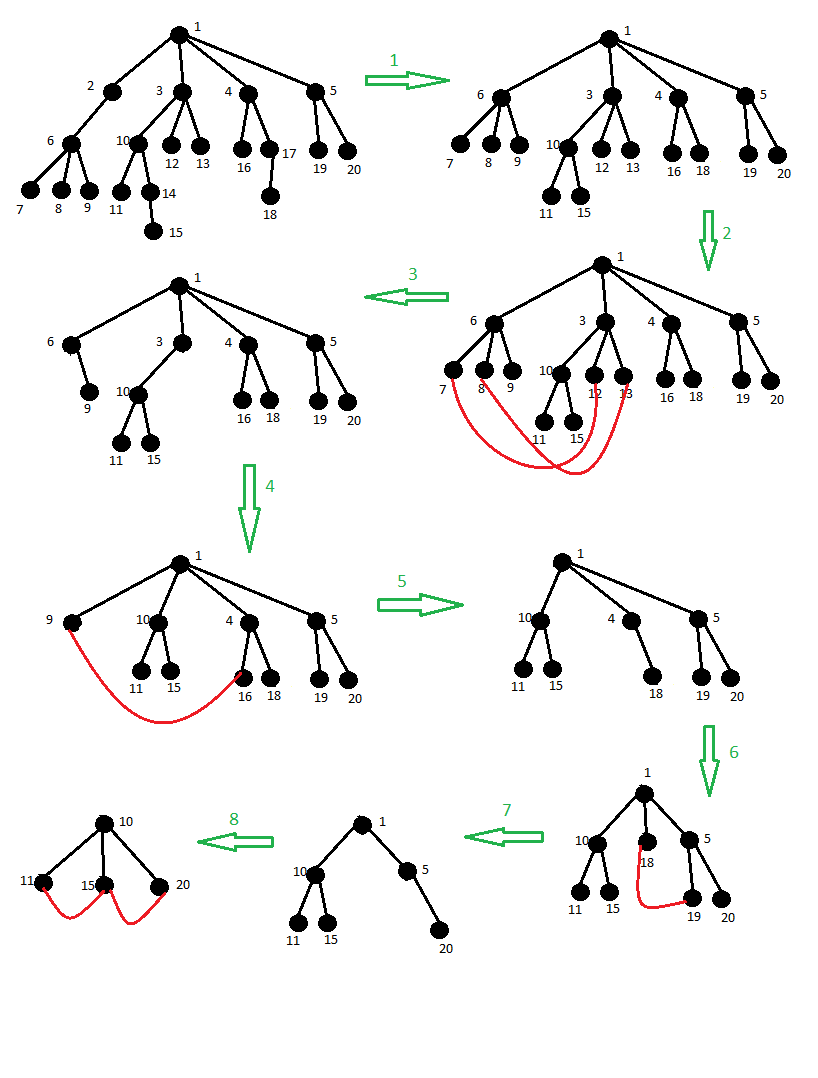}
\caption {Trace of Parallel Bi-connectivity Augmentation Algorithm}   \label{fig:trace}
\end{center}
\end{figure}
\begin{enumerate}
\item We first perform path-to-edge contraction on the input tree.  For example, the path $P_{1~6}$ gets contracted to the edge $\{1,6\}$. Similarly paths $P_{10~15}$ and $P_{4~18}$ gets contracted to respective edges.
\item The representative set is $\{3, 4, 5, 6, 10\}$ with degrees $\{4, 3, 3, 4, 3\}$, respectively.  For the current iteration, nodes 3 and 6 are the candidate representatives as its degree are maximum and second maximum.  The number of leaf nodes in node 3 is 2 and the number of leaf nodes in node 6 is 3.  So add 2 edges $\{7,12\}$ and  $\{8,13\}$.
\item Delete the nodes 7, 12, 8 and 13 from the tree.
\item If required, perform path-to-edge contraction.  Update the sets $R$ and $D$.  The representative nodes are 1, 4, 5 and 10 with degrees 4, 3, 3, 3, respectively.  The maximum degree node is 1 with degree 4 and the number of leaf nodes associated with it is 1.  The second maximum degree node is 4 with degree 3 and the number of leaf nodes is 2. So, 1 edge gets added from 9 to 16.
\item Nodes 5 and 16 gets deleted.
\item Contraction takes place converting path to a single edge.  The representative nodes are 1, 5, 10 with degree 3 each.  Maximum degree node is 1 with the number of leaf nodes as 1 and the second maximum degree node is 5 with the number of leaf nodes as 2. So, 1 edge gets added from node 18 to node 19.
\item Node 18 and 19 gets deleted.
\item The resultant tree gets contracted. The path joining 10 to 15 gets contracted to a single edge. The resultant tree is a star tree with node 10 of degree 3 with leaf nodes as 11, 15 and 20.  In star-augment, 2 edges, namely $\{11,15\}$ and $\{11,20\}$ are added to the tree.  The algorithm is complete and the resulting 2-connected graph is shown in the  Figure ~\ref{fig:out}. 
\end{enumerate}
\begin{figure}[htpb]
\begin{center}
\includegraphics[angle=0, scale=0.7]{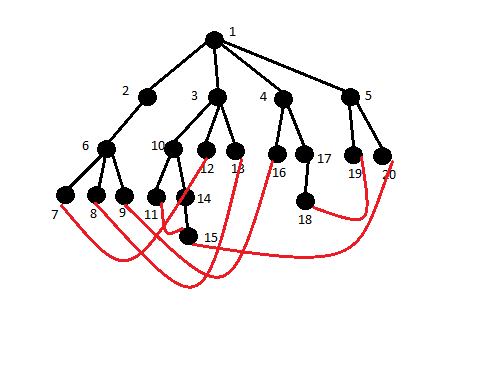}
\caption {The Resulting 2-connected Graph}   \label{fig:out}
\end{center}
\end{figure}
\begin{lemma}
Parallel biconnectivity augmentation algorithm is optimal with respect to $\Delta$ processors. 
\end{lemma}
\begin{proof}
We analyze the cost incurred for each step in the algorithm.  
\begin{itemize}
\item {\bf Edge Contraction:} For a given input tree, our path-to-edge contraction routine will be called at most $(n-\Delta -1)$ times.  Since edge contraction is a constant time effort, the time complexity for parallel edge contraction routine is $O(n-\Delta-1)$, which is $O(n)$.
\item If $\Delta$ processors are available.  These edges can be divided equally among $\Delta$ processors. Hence, the processor time product is $O(\Delta * (n-\Delta-1)/\Delta)$
\item {\bf Parallel non-star-augment:} Hash tables $H1$ and $H2$ are populated using a single processor following which augmentation takes place in parallel. The $t$ edges determined by our algorithm are added in parallel.  Since $t \leq \Delta$, each iteration augments at most $\Delta$ edges and this routine is called at most $O(\frac{n}{\Delta})$ times.  Hence, the processor-time product is $O(\Delta \times \frac{n}{\Delta})= O(n)$.  Over all iterations, the cost incurred in updating $H1$ and $H2$ is $O(n)$.
\item If the tree becomes a star-like tree, then we need to augment at most $\Delta$ edges which can be done using $\Delta$ processors.
\item Therefore, the overall time complexity is $O(n)$ with respect to $\Delta$ processors.
\end{itemize}
\end{proof}
\subsection{Simulation of Parallel Biconnectivity Augmentation Algorithm}
In this section, we present our simulation results of Algorithm \ref{tree-alg5}.  The algorithm has been simulated using Dual core processor and simulation time for various input trees is shown in the Figure \ref{simulation}.
\begin{figure}[htpb]
\begin{center}
\includegraphics[angle=0, scale=0.65]{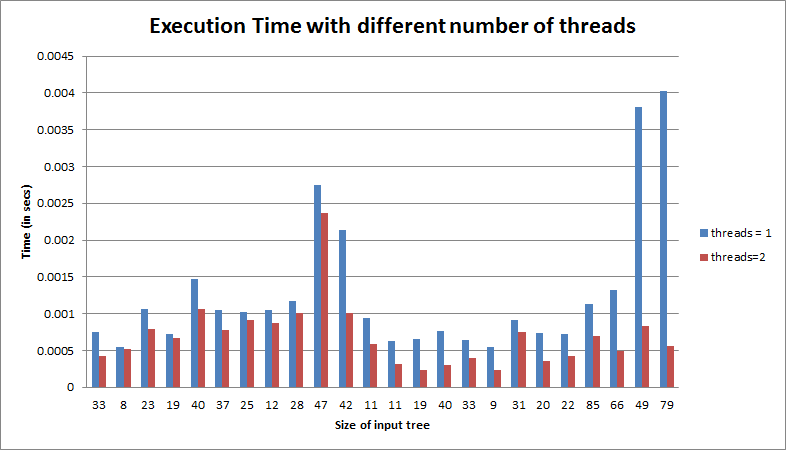}
\caption{Comparison between Execution times of Sequential and Parallel Augmentation Algorithm} \label{simulation}
\end{center}
\end{figure}
\subsection{Applications}
In \cite{nsn}, the biconnectivity augmentation of 1-connected graphs is reported using the biconnected component tree (bcc-tree) and triconnectivity augmentation of 2-connected graphs is reported using the 3-block tree.  We highlight that using bcc-tree, our proposed framework gives both sequential and parallel algorithm for biconnectivity augmentation of 1-connected graphs.  Similarly, our framework yields parallel triconnectivity augmentation of 2-connected graphs with the help of 3-block trees.
\section{Conclusion and Future Directions}
In this paper, we have presented a new perspective of biconnectivity augmentation algorithm using edge contraction tool.  Our approach simplifies the results reported in \cite{nsn}.  Using our new sequential approach, we have also presented a parallel biconnectivity augmentation algorithm for trees.  Our parallel algorithm is optimal with respect to $\Delta$ processors, where $\Delta$ is the maximum degree of the tree.   Since the edge contraction preserves all connectivity information,  we believe that this tool may be of use in other combinatorial problems.  An interesting direction for further research is to use this framework to solve general connectivity augmentation problem which asks for increasing the vertex connectivity to $r \geq 2$, given a graph with vertex connectivity $k \geq 1$.


\begin{thebibliography}{4}
\bibitem{JaJa}
Joseph Ja Ja, An Introduction to Parallel Algorithms,  Addison Wesley, 1992.
\bibitem{tarjan}
K.P.Eswaran, R.E.Tarjan: Augmentation problems. SIAM Journal of Computing, 5, 653-665 (1976)
\bibitem{hsu}
T.S.Hsu, V.Ramachandran: On finding a smallest augmentation to biconnect a graph. SIAM Journal of computing, 22, 889-912 (1993)
\bibitem{vegh}
L.A.Vegh: Augmenting undirected node-connectivity by one. In Proceedings of the 42nd ACM Symposium on Theory of Computing (STOC) (2010)
\bibitem{karger}
R.Karger: Using randomized sparsification to approximate minimum cuts. In Proceedings of the 5th ACM-SIAM Symposiumon Discrete Algorithms, 1994. 
\bibitem{T}
Tarjan, R. and Vishkin, U. An Efficient Parallel Biconnectivity Algorithm, SIAM Journal on Computing, 1985, Vol. 14.
\bibitem{nsn}
N.S.Narayanaswamy and N.Sadagopan: A Novel Data Structure for Biconnectivity, Triconnectivity, and k-Tree Augmentation. In Proc. Computing: The Australasian Theory Symposium (CATS 2011). CRPIT, 119, ACS 45-54. 
\bibitem{west}
Douglas B. West: Introduction to Graph Theory, 2nd Edition, 2000.
\bibitem{golu}
M.C.Golumbic: Algorithmic graph theory and perfect graphs, Academic Press. (1980)
\bibitem{easwar}
K.P.Eswaran and R.E.Tarjan (1976), Augmentation problems, SIAM Journal of Computing, 5, 653-665.
\bibitem{a} T.S.Hsu, V.Ramachandran: A linear-time algorithm for triconnectivity augmentation. In Proc. of 32nd Annual IEEE Symp. on Foundations of Comp. Sci.(FOCS), pp.548-559 (1991)
\bibitem{b} T.S.Hsu: On four connecting a triconnected graph. Journal of Algorithms, 35, 202-234 (2000)
\bibitem{matthias} M.Kriesell: A Survey on Contractible Edges in Graphs of a Prescribed Vertex Connectivity, Graphs and Combinatorics, 18, 1-30, (2002). 
\end{thebibliography}
\end{document}